\documentclass[11pt,reqno,a4paper]{article}
\usepackage{amsmath}
\usepackage{amsbsy}
\usepackage{amsthm}
\usepackage{amssymb}
\usepackage{amscd}
\usepackage{mathabx}
\usepackage{accents}
\usepackage{float}
\usepackage{url}
\usepackage{tikz}
\usetikzlibrary{matrix,arrows,decorations.pathmorphing}
 \usepackage{relsize}

\usepackage{yfonts}
\usepackage{amssymb}
\usepackage{amsmath}

\usepackage{mathabx}

\usepackage{hyperref}
\hypersetup{
pdftitle={},%
pdfauthor={},%
pdfsubject={},%
pdfkeywords={},%
colorlinks=true,%
linkcolor={black},%
linktoc={},
linktocpage={},%
pageanchor={},
citecolor={black},
}

\usepackage[english]{babel}
\usepackage[utf8]{inputenc}
\usepackage[margin=2.41cm]{geometry}
\usepackage{mathrsfs}
\usepackage[all]{xy}
\usepackage{hyperref}
\usepackage{url}

\usepackage{cite}

\usepackage{titlesec}
\titleformat{\section}{\large\bfseries\filcenter}{\thesection}{1em}{}
\titleformat{\subsection}{\bfseries}{\thesubsection}{1em}{}

\usepackage{etoolbox}

\makeatletter
\patchcmd{\ttlh@hang}{\parindent\z@}{\parindent\z@\leavevmode}{}{}
\patchcmd{\ttlh@hang}{\noindent}{}{}{}
\makeatother
\newtheorem{thm}{Theorem}[section]
\newtheorem{cor}[thm]{Corollary}
\newtheorem{lemma}[thm]{Lemma}

\theoremstyle{remark}

\theoremstyle{definition}

\newtheorem{rmk}[thm]{Remark}
\newtheorem{defn}[thm]{Definition}
\newtheorem{example}[thm]{Example}

\numberwithin{equation}{section}
\allowdisplaybreaks[1]

\catcode`@=12

\renewcommand\thanks[1]{%
  \begingroup
  \renewcommand\thefootnote{}\footnote{#1}%
  \addtocounter{footnote}{-1}%
  \endgroup
}

\renewcommand{\tilde}{\widetilde}
\renewcommand{\epsilon}{{\varepsilon}}

\usepackage{bbm}

\def\bearray{\begin{eqnarray}}
\def\earray{\end{eqnarray}}
\def\beq{\begin{equation}}
\def\eeq{\end{equation}}

\def\b0{{\bf 0}}

\def\R{\Omega}

\def\cA{{\cal A}}
\def\cB{{\cal B}}

\def\cR{{\cal R}}

\newsymbol\bt 1202             %%% \boxtimes 

\def\bC{{\mathbb C}}           %%%  complex numbers and so on 

\def\bR{{\mathbb R}}

\newsymbol\rest 1316         %%% restriction symbol 
\def\C{{\mathbb C}}

\def\N{{\mathbb N}}

\def\R{{\mathbb R}}
\def\S{{\mathbb S}}

\def\hotimes{\hat{\otimes}_\epsilon}

\begin{document}

\begin{flushright}

\baselineskip=4pt

\end{flushright}

\begin{center}
\vspace{5mm}

{\Large\bf INJECTIVE TENSOR PRODUCTS \\[3mm] IN STRICT DEFORMATION QUANTIZATION}

%\thanks{}

\vspace{5mm}

{\bf by}

\vspace{5mm}
\noindent
{  \bf Simone Murro}\\[1mm]
%\noindent  {\it Dipartimento di Matematica, Universit\`a di Trento and INFN-TIFPA  I-38123 Povo, Italy}\\[1mm]
\noindent  {\it Laboratoire de Math\'ematiques d’Orsay, Universit\'e Paris-Saclay, 
91405 Orsay,  France}\\[1mm]
email: \ {\tt simone.murro@u-psud.fr}
\\[6mm]

 { \bf Christiaan J.F.  van de Ven}\\[1mm]
\noindent  {\it Marie Sk\l odowska-Curie Fellow of the Istituto Nazionale di Alta Matematica}\\[1mm]
\noindent  {\it Dipartimento di Matematica, Universit\`a di Trento and INFN-TIFPA  I-38123 Povo, Italy}\\[1mm]
email: \ {\tt christiaan.vandeven@unitn.it}
\\[10mm]
\end{center}

\begin{abstract}
The aim of this paper is two-fold. Firstly we provide necessary and sufficient criteria for the existence of a strict deformation quantization of algebraic tensor products of Poisson algebras, and secondly we discuss the existence of products of KMS states.

As an application, we discuss the correspondence between quantum and classical Hamiltonians in spin systems and we provide a relation between the resolvent of Sch\"odinger operators for non-interacting many particle systems and quantization maps.
\end{abstract}

resolvent algebras, quantum spin system, Heisenberg model, Ising model, Curie-Weiss model.
\\[0.5mm]
%\tableofcontents

%%%%%%%%%%%%%%%%%%%%%%%%%%%%%%%%%%%%
\section{Introduction}
The concept of {\em strict deformation quantization} has been introduced by Rieffel in\cite{Rie89} in order to provide a mathematical formalism that describes the transition from a classical theory to a quantum theory in terms of deformations of (commutative) Poisson algebras (representing the classical theory) into non-commutative $C^*$-algebras (characterizing the quantum theory). 
More precisely, given a commutative $C^*$-algebra $A_0$ the strict deformation quantization of $A_0$ consists of the assignment of a continuous bundle  $\cA$ of $C^*$-algebras $(A_\hbar)_{\hbar\in I}$ over an interval $I$ along with a family of quantization maps $Q_\hbar:\tilde{A}_0\to A_\hbar$, with $\hbar \in I$ and $\tilde{A}_0\subset A_0$ a dense Poisson subalgebra of $A_0$, which rules the deformation of $A_0$ (cf. Definition~\ref{def:deformationq}). 
Once that a quantum theory is constructed, the classical counterpart is obtained by performing the so-called \emph{classical limit}, i.e. $\hbar\to 0$ (see \cite{Lan17,MV2,Ven21} for a rigorous construction). For sake of completeness, let us illustrate this with an example of the strict deformation quantization of a classical particle on the phase space $\R^{2n}$.

\paragraph{Quantization of a classical particle} 
The classical observables of a free particle on the phase space $\R^{2n}$ are encoded in the ring of continuous functions vanishing at infinity on this space, i.e. $C_0(\mathbb{R}^{2n})$, which in particular contains (a) commutative dense Poisson algebra(s). For convenience we take the simplest functional-analytic setting in which only smooth compactly supported functions $f\in C_c^{\infty}(\mathbb{R}^{2n})$ (with Poisson structure given by the natural symplectic form $\sum_{j=1}^ndp_j\wedge dq^j$) are quantized. In order to relate $C_c^{\infty}(\mathbb{R}^{2n})$ to a quantum theory described on some Hilbert space, one needs to deform $C_c^{\infty}(\mathbb{R}^{2n})$ into non-commutatative $C^*$-algebras exploiting a family of quantization maps. In this setting the family of quantization maps are given by
\begin{align*}
&Q_{\hbar}:  C_c^{\infty}(\mathbb{R}^{2n}) \to B_{\infty}(L^2(\mathbb{R}^n));\\
&Q_{\hbar}(f)=\int_{\mathbb{R}^{2n}} \frac{d^npd^nq}{(2\pi\hbar)^n} f(p,q) |\phi_{\hbar}^{(p,q)}\rangle \langle\phi_{\hbar}^{(p,q)}|,
\end{align*}
 where  $\hbar\in (0,1]$, $B_{\infty}(\mathcal{H})$ is the $C^*$-algebra of compact operators on the Hilbert space $\mathcal{H}=L^2(\mathbb{R}^n)$ with the usual Lebesgue measure $d^npd^nq$ and, for each point $(p,q)\in \mathbb{R}^{2n}$, the operator $|\phi_{\hbar}^{(p,q)}\rangle \langle \phi_{\hbar}^{(p,q)}|:L^2(\bR^n)\to L^2(\mathbb{R}^n)$ is defined as the orthogonal projection 
onto the linear span of the normalized wavefunctions $\phi_{\hbar}^{(p,q)}$ given, for $x\in\mathbb{R}^n$, by
\begin{align}
\phi_{\hbar}^{(p,q)}(x)=(\pi\hbar)^{-n/4}e^{-ipq/2\hbar}e^{-ipx/\hbar}e^{-(x-q)^2/2\hbar}\:, \quad \phi_{\hbar}^{(p,q)}\in L^2(\mathbb{R}). \label{schrcoherent}
\end{align}
The functions~\eqref{schrcoherent} are dubbed (Schr\"odinger) {\em coherent states}. In~\cite{Rie89,Rie94} Rieffel showed that the fibers  $A_0=C_0(\mathbb{R}^{2n})$, and $A_{\hbar}=B_{\infty}(\mathcal{H}) \ (\hbar\in (0,1])$ can be combined into a (locally non-trivial) {\em continuous bundle $A$ of $C^*$-algebras} over base space $I = [0, 1]$; the maps $Q_{\hbar}$ which are defined on the dense subspace $C_c^{\infty}(\mathbb{R}^{2n})\subset A_0$ are called {\em quantization maps}.\\ \medskip

As noticed by Landsman in\cite{Lan98,Lan17}, a continuous bundle of $C^*$- algebras provides a natural setting to describe models in quantum statistical mechanics. By interpreting the  semi-classical parameter as the number of particles of a system, namely $\hbar=1/N \in 1/\N\cup\{0\}$, the limit $N\to \infty$ provides the so-called {\em thermodynamic limit}, namely the density of the system $N/V$ is kept fixed, and the volume $V$ of the system sent to infinity, as well. This has been rigorously studied using operator algebras since the 1960s.  The limiting system constructed at the limit $N=\infty$ is typically quantum statistical mechanics in infinite volume.  In this setting the so-called {\em quasi-local} observables are studied: these give rise to a non-commutative continuous bundles of $C^*$-algebras, namely $A^{(q)}$, defined over the base space $I:=1/\N\cup \{0\}\subset[0,1]$ with fibers at $1/N$ given by a $N$-fold tensor product of a matrix algebra with itself. 
However, the limit $N\to\infty$ can also provide the relation between classical (spin) theories viewed as limits of quantum statistical mechanics. In this case the  {\em quasi-symmetric} (or {\em macroscopic}) observables are studied and these induce a commutative bundle of $C^*$-algebras denoted by $A^{(c)}$  which is defined over the same base space $I:=1/\N\cup \{0\}\subset[0,1]$ with exactly the same fibers at $1/N$ as the algebra $A^{(q)}$, but differ at $N=\infty$, i.e., $1/N=0$. 
 It is precisely the bundle $A^{(c)}$ which relates these (spin) systems to strict deformation quantization, since macroscopic observables are defined by (quasi-) symmetric sequences which in turn are induced by certain quantization maps. Again, these maps can be used to prove the existence of the classical limit for quantum spin systems which has particularly been done for mean-field quantum spin systems \cite{LMV,Ven20}.\medskip

As noticed for the first time by Rieffel in \cite{Rie89} non-commutative tori can be considered as a strict deformation quantization of ordinary tori with an appropriate Poisson structure. As a consequence it is reasonable to expect that any symplectic twisted group $C^*$-algebra (see e.g.~\cite{mu1,mu2}) can be seen as a strict deformation of ordinary manifold. But it is not clear if any ordinary (Poisson) manifold does admit a strict deformation quantization and having a general criterion for the existence of a strict deformation quantization still seems to be too far reaching.
 Let us remark that noncommutative geometry has many interesting applications in physical theory, like the quantum hall effect (see e.g.~\cite{qhe}) and abelian Chern-Simons theory (see e.g.~\cite{mu0}).
The aim of this paper is dual: on the one hand, we shall provide a sufficient criterion for the existence of a strict deformation quantization of algebraic product of Poisson algebras (cf. Theorem~\ref{thm:main}). On the other hand, we shall prove that  the products of KMS states is still a KMS state (cf. Theorem~\ref{thm:kms}). As a direct consequence of Theorem~\ref{thm:main} we show that given two locally compact Poisson manifold $X$ and $Y$, which admit strict deformation quantization over the interval $I=1/\N \cup\{0\}$, also the Poisson manifold $X\times Y$ does so (cf. Corollary~\ref{cor:XxY}). 
\medskip

The paper is structured as follows. In the 2nd section, we fix our notation and we recall some results from the theory of operator algebras.  Section 3 and 4 are the core of the paper where the main result are obtained.  Finally in Section 5  we discuss some applications of our main results to spin systems and resolvent algebra.

\subsection*{Acknowledgments}
We are grateful to Federico Bambozzi, Francesco Fidaleo, Klaas Landsman, Valter Moretti and Teun van Nuland for helpful discussions related to the topic of this paper. We are grateful to the referee for the useful comments.

\subsection*{Funding}
S.M is supported by the the DFG research grant MU 4559/1-1 ``Hadamard States in Linearized Quantum Gravity'' and thanks the support of the INFN-TIFPA project ``Bell'' and the University of Trento during the initial stages of this project. The second author is funded by the `INdAM Doctoral Programme in Mathematics and/or Applications' co-funded by Marie Sklodowska-Curie Actions, INdAM-DPCOFUND-2015, grant number 713485.

\subsection*{Statements and Declarations}

The authors have no competing interests to declare that are relevant to the content of this article.

\section{Preliminaries}\label{sec:preliminaries}
In this section we collect the basic facts and conventions concerning operator algebras and strict deformation quantization of  Poisson algebras. For a detailed introduction the reader may consult~\cite{Lan98,Lan17,OA1}.\medskip

\subsection{The injective tensor product of continuous bundles of $C^*$-algebras}\label{sec:prel inj ten}

In this section, we shall collect basic facts about injective tensor products of continuous bundles of $C^*$-algebra. We begin by recasting the definition of continuous bundle of $C^*$-algebras.

\begin{defn}\label{def:continuous algebra bundle}
A \emph{bundle of} $C^*$-\emph{algebras} over a locally compact Hausdorff space $I$ is a triple $\cA:=(I,A,\pi_\hbar:A\to A_\hbar)$, where $A$ is a $C^*$-algebra (the bundle $C^*$-algebra) and, for each $\hbar \in I$, $\pi_\hbar$ is a
$*$-epimorphism of $A$ onto a $C^*$-algebra $A_\hbar$ such that:
\begin{itemize}
\item[(i)] the family $\{\pi_\hbar | \hbar \in I \}$ is faithful, i.e. $\|a \|= \sup_{\hbar\in I} \|\pi_\hbar(a)\|_\hbar$  for each $\hbar\in I$ and $\|\cdot\|$ (resp. $\|\cdot\|_\hbar$) denote the $C^*$-norm of $A$ (resp. $A_\hbar$);
\item[(ii)] there exist an action $\rho: C_0(I) \times A \to A$ such that $\pi_\hbar(\rho(f , a)) = f(\hbar) \pi_\hbar(a)$ for any $\hbar\in I$.
\end{itemize} 
A continuous bundle of $C^*$-algebras is a $C^*$-bundle $\cA=(I,A,\pi_\hbar)$ which also satisfies
\begin{itemize}
\item[(iii)] for $a \in A$, the norm function $N(a): \hbar \mapsto \|\pi_\hbar(a)\|_\hbar$ is in $C_0(I)$. 
\end{itemize}
A continuous {\bf section} of the bundle is an element $\{a_{\hbar}\}_{\hbar\in I}$ of $\Pi_{\hbar\in I}A_\hbar$ for which there exists an $a\in A$ such that $a_\hbar=\pi_\hbar(a)$ for each $\hbar\in I$.
It is not requested that the $C^*$-algebras $A_\hbar$ are unital. 
If all the $A_\hbar$ are instead unital, then  also $A$ is assumed to be unital and $\pi_\hbar$ is supposed to be 
unit-preserving.
\end{defn} 

\begin{rmk}\label{pointwise}

Notice that, since the $\pi_\hbar$ are homomorphisms of $C^*$-algebras,  the $*$-algebra operations in $A$ correspond to the corresponding pointwise operations of the sections $I\ni \hbar \mapsto \pi_\hbar(a)$. Condition (ii) reinforce the linearity preservation condition permitting coefficients continuously depending on $\hbar$.
\end{rmk}

As explained in the introduction of ~\cite{Kirchberg-Wassermann}, Definition~\ref{def:continuous algebra bundle} is equivalent to the classical definition of a continuous field of $C^*$-algebras ~\cite[Definition 10.3.1]{Dixmier}. 
%\textcolor{red}{Indeed we can identify $A$ with the sections of the bundle: if we do so, the homomorphism $\pi_\hbar$ is just the evalution map at $\hbar$.} 
Indeed we can identify $A$ with the $*$-algebra of elements $\gamma$ in the cartesian product $\Pi_{\hbar \in I}A_\hbar$ for which there is an $a \in A$ with $\gamma_\hbar = \pi_\hbar(a)$ for $\hbar\in I$.
If $\Gamma$ is the $*$-algebra of elements of $\Pi_{\hbar \in I} A_\hbar$ which coincide on compact subsets of $I$ with elements of $A$, the triple $(I, A, \pi_\hbar)$ is a continuous field of $C^*$-algebras in the sense of \cite{Dixmier}, and the subset of continuous functions vanishing at infinity $C_0(\Gamma) $ equals $A$. Conversely, if $(I, A, \pi_\hbar)$ is a continuous field of $C^*$-algebras on $I$ and $A$ is the $*$-algebra of $\gamma\in\Gamma$ 
such that the function $\hbar\mapsto \|\gamma_\hbar\|$ is in $C_0(I)$, then $A$ is a $C^*$-algebra and
$(I, A, \pi_\hbar: A \to A_\hbar)$ is a
continuous bundle in the sense of Definition~\ref{def:continuous algebra bundle}, with $A = C_0(\Gamma)$. \medskip

If $\cA$ and $\cB$ are continuous bundles of $C^*$-algebras there exists a natural bundle $\cA\otimes \cB$ over $I$ with bundle algebras given by the algebraic tensor product $A\otimes B$. Clearly $\cA\otimes \cB$ is not a bundle of $C^*$-algebras since the algebraic tensor product $A\otimes B$ is only a pre-$C^*$-algebra. Therefore, a suitable completion of $A\otimes B$ has to be performed to obtain a $C^*$-algebra.
A natural strategy is to embed $A\otimes B$ as a $*$-subalgebra of algebra of bounded operators $B(\mathcal{H})$ for some Hilbert space $\mathcal{H}$: The norm of an element in $A\otimes B$ will then be the
operator norm of the associated bounded operator. The resulting norm on $A\otimes B$ is usually dubbed \emph{injective tensor norm} (or \emph{spatial norm} or \emph{minimal $C^*$-norm}) and we will denote it as $\|\cdot\|_\epsilon$.
We summarize the above discussion in the following theorem and we refer to \cite{OA1} for more details.
\begin{thm}[\protect{\cite[Theorem B.9]{OA1}}]\label{thm:inject tens norm}
Let $A$ and $B$ be $C^*$-algebras and consider two faithful representations $\pi_A: A \to B(\mathcal{H}_A)$ and $\pi_B:B \to B(\mathcal{H}_B)$. Then it holds:
\begin{itemize}
\item[-] There exists a unique $*$-homomorphism $\pi_A\otimes \pi_B: A\otimes B \to B(\mathcal{H}_A \otimes \mathcal{H}_B)$ such that $\pi_A\otimes \pi_B(a\otimes b)=\pi_A(a)\otimes \pi_B(b)$;
\item[-] The  $C^*$-norm $\|\cdot\|_\epsilon$ on $A\otimes B$  defined by 
$$\|\sum_{i=1}^k a_i\otimes b_i \|_\epsilon := \|\sum_{i=i}^k \pi_A(a_i)\otimes \pi_B(b_i) \|_{B(\mathcal{H}_A \otimes \mathcal{H}_B)} $$ 
 does not depend on the
choice of representations and it is a cross-norm, i.e. for all $a_i\in A$ and $b_i\in B$ it holds 
\begin{equation}
\label{eq:cross norm}\|a_i\otimes b_i\|_\epsilon = \|a_i\|_A \|b_i\|_B\,
\end{equation}
 where $\|\cdot\|_A$ and $\|\cdot\|_B$ are the $C^*$-norm of $A$ and $B$ respectively.
\end{itemize}   

\begin{defn}\label{def:inject tens prod}
Given two $C^*$-algebras $A$ and $B$, we call \emph{injective tensor product} of $A$ and $B$  the completion $A\hotimes B$ of $A\otimes B$ with respect to the injective tensor norm $\|\cdot\|_\epsilon$.
\end{defn}
\end{thm}
\begin{example}\label{ex:inject tensor prod}
There are some basic examples where the injective tensor product of two $C^*$-algebras takes a familiar form. When one algebra is commutative, for example, we can identify the injective tensor product with an algebra of complex-valued functions. If $X$ is a locally compact Hausdorff space and $A$ is a $C^*$-algebra, then the ring  $C_0(X, A)$ of continuous functions 
$f : X \to A$  such that $x \mapsto \|f(x)\|$ vanishes at infinity is a $C^*$-algebra with pointwise operations and the supremum norm:
$$fg (x)=f(x)g(x) \qquad f^*(x)=f(x)^* \qquad \|f\|_0=\sup_{x\in X} \|f(x)\| \,.$$
As shown in \cite[Corollary B.17,]{OA1} if $X$ and $Y$ are locally compact Hausdorff spaces, then there is an isomorphism $\psi$ of $C_0(X)\hotimes C_0(Y)$ onto $C_0(X\times Y)$ such that $\psi(f\otimes g)(x,y) = f(x)g(y)$ for every $f\in C_0(X)$ and $g\in C_0(Y)$.
\end{example}
\noindent Replacing the algebraic tensor product $\cA\otimes \cB$ with the injective tensor product $\cA\hotimes\cB$, we thus obtain a bundle of $C^*$-algebras but this bundle is only lower-semicontinuous as shown by Kirchberg and Wasserman in~\cite[Proposition 4.9]{Kirchberg-Wassermann}.  A sufficient criterium for continuity is obtained in \cite[Remark 2.6.1]{Kirchberg-Wassermann}, by combining \cite[Lemma 2.4 and 2.5]{Kirchberg-Wassermann}. We recall the result for sake of completeness.
\begin{lemma}[\protect{\cite[Remark 2.6.1]{Kirchberg-Wassermann}}]\label{lem:nucl}
Let $\cA = (I,A, \pi_\hbar:A\to A_\hbar)$ and $\cB = (I,B, \sigma_\hbar:B\to B_\hbar)$ be
continuous bundles of $C^*$-algebras. If for every $\hbar\in I $ the algebras $A_\hbar$ and $B_\hbar$ are nuclear $C^*$-algebras, then $\cA \hotimes \cB$ is a continuous bundle of $C^*$-algebras.
\end{lemma}
%\begin{proof}
%Let denote by $\otimes_{max}$ the projective tensor product and by $\|\cdot\|_{max}$ the projective tensor norm. Since $A_\hbar$ and $B_\hbar$ are a nuclear $C^*$-algebra,  the completion of $\cA\otimes\cB$ with respect the injective tensor norm $\|\cdot\|_\epsilon$ and the projective tensor norm $\|\cdot\|_{max}$ coincides as well as their norms, namely $\cA\hat{\otimes}_{max}\cB = \cA\hotimes\cB$ and $\|\cdot\|_\epsilon=\|\cdot\|_{max}$, (see e.g.~\cite[Section II.9.4]{OA2} for more details).
%Therefore, to conclude our proof, we need to show that the bundle $\cA\hotimes\cB$ is continuous. 
%But, this follows immediately from the fact that $\hbar \mapsto \|(\pi_\hbar\otimes_\epsilon {\rm Id})(\cdot)\|_\hbar$ is lower and upper semicontinuous (see e.g.~\cite[Lemma 2.4 and Lemma 2.5]{Kirchberg-Wassermann}). 
%\end{proof}
\begin{rmk}
Clearly assuming that $A_\hbar$ and $B_\hbar$ are nuclear is a sufficient but not a necessary condition. On account of \cite[Theorem 4.6]{Kirchberg-Wassermann} one can even take one bundle to nuclear.
\end{rmk}
%Clearly assuming that $A_\hbar$ and $B_\hbar$ are nuclear is a sufficient but not a necessary condition.
A sufficient and necessary condition however was provided by Archbold in \cite{Archbold}.
\begin{thm}[\protect{\cite[Theorem 3.3]{Archbold}}]\label{Archbold}
Let $\cA = (I,A, \pi_\hbar:A\to A_\hbar)$ and $\cB = (I,B, \sigma_\hbar:B\to B_\hbar)$ be
continuous bundles of $C^*$-algebras. Then for each $\hbar \in I$, the function $\hbar \mapsto \|(\pi_\hbar \otimes \sigma_\hbar )(c)\|_\hbar$ is continuous for all $c\in A\hotimes B$ at $\hbar$   if and only if 
$$\ker(\pi_\hbar \otimes \sigma_\hbar ) = \ker(\pi_\hbar ) \hotimes B + A \hotimes \ker(\sigma_\hbar)\,.$$ 
\end{thm}

\subsection{Strict deformation quantization}

A Poisson algebra is a real (or complex) algebra endowed with a Poisson bracket, i.e. a skew-symmetric bilinear map $\{\cdot, \cdot\}:A\times A \to A$ which satisfies Jacobi identity and Leibniz rule. 
If the algebra is endowed with an involution, i.e. $A$ is a $*$-algebra, we additional demand that, for every $f,g\in A$, it holds $\{f, g\}^* = \{f^* , g^* \}$. We now give the definition of a strict deformation quantization. %The one we use is equivalent to \cite{Lan98}[Definition 1.1.1]. 

\begin{defn}\label{def:deformationq}
A {\it strict deformation quantization} of a Poisson algebra $\tilde{A}_0$ densely contained in a commutative $C^*$-algebra $A_0$ consists of:
 \begin{itemize}
\item[(I)]  A {continuous  bundle of unital $C^*$-algebras} $\cA:=(I,A,\pi_\hbar:A\to A_{\hbar})$,  (with norms $||\cdot||_{\hbar}$) where $I$ is an subset of $\mathbb{R}$ containing $0$ as accumulation point;
 \item[(II)]  A collection of  linear $*$-preserving {\it quantization maps}, namely a family $Q:=\{Q_\hbar\}_{\hbar\in I}$ of maps $Q_{\hbar}:\tilde{A}_0  \to A_{\hbar}$ such that: 
\begin{enumerate}
%\item[(i)] $Q_{\hbar}(\tilde{A}_0)$ is a dense $*$-subalgebra of $A_{\hbar}$ (for each $\hbar\in I$).
\item[(i)] $Q_0$ is the inclusion map $\tilde{A}_0 \hookrightarrow A_0$ and $Q_{\hbar}(\mathrm{1}_{A_0})=\mathrm{1}_{A_{\hbar}}$ (the unit of $A_{\hbar}$);
\item[(ii)] Each $Q_{\hbar}$ is self-adjoint,  i.e. $Q_{\hbar}({f}^*) = Q_{\hbar}(f)^*$;
\item[iii)] For each $f\in\tilde{A}_0$ the following cross-section of the bundle is continuous:
\begin{align*}
&0\mapsto f;\\
&\hbar\mapsto Q_{\hbar}(f), \  \ (\hbar\in (I\setminus\{0\}));
\end{align*}
\item[(iv)]  Each pair $f,g\in \tilde{A}_0$ satisfies the Dirac-Groenewold-Rieffel condition:
\begin{align*}
\lim_{\hbar\to 0}||\frac{i}{\hbar}[Q_{\hbar}(f),Q_{\hbar}(g)]-Q_{\hbar}(\{f,g\})||_{\hbar}=0.
\end{align*}
\end{enumerate}
\end{itemize}
\end{defn}

\begin{rmk}
Notice that Definition~\ref{def:deformationq} generalizes the classical definition of strict deformation quantization of a Poisson manifold $X$ (see e.g.\cite[Definition 7.1]{Lan17}). Indeed, once that a Poisson structure is defined on a dense $C^*$-subalgebra            $\tilde{A}_0$  of the algebra of continuous functions vanishing at infinity $C_0(X)$, it is easy to check that $A_0:=C_0(X)$ is a $C^*$-algebra with $C^*$-norm given by the supremum norm.
\end{rmk}

\begin{rmk}
If one requires the quantization maps $Q_{\hbar}$ to be injective for each $\hbar$ and that $Q_{\hbar}(\tilde{A}_0)$ is a dense $*$-subalgebra of $A_{\hbar}$ (for each $\hbar\in I$), then the previous definition defines a {\em strict deformation quantization} in the sense of \cite[Definition 1.1.2]{Lan98}.  If one requires that the base space $I$ is discrete or such that  $A_{\hbar}$ are identical for each $\hbar\neq0$ then the quantization maps in Definition \ref{def:deformationq} uniquely define this bundle \cite[Theorem 1.2.4]{Lan98}.
\end{rmk}

%With the next proposition, we provide a criteria to construct a continuous bundle of $C^*$-algebras once that a family of $C^*$-algebras is assigned. 
%\begin{prop}[\cite{Lan17}, Proposition C.124]\label{prop c124}
%Let $\{{A}_{\hbar}\}_{\hbar}$ be a family of $C^*$-algebras over the base space $I=1/\N \cup \{0\}\subset [0,1]$ and denote with $\|\cdot\|_\hbar$ the associated $C^*$-norm of $ {A}_\hbar$. Assume there exists a $*$-algebra $\tilde{A}\subset\Pi_{\hbar} {A}_{\hbar}$ and  $*$-homomorphisms $\pi_\hbar: \tilde{A}\to  {A}_\hbar$ such that, 
%\begin{itemize}
%\item[(i)] for each $\hbar\in I$, $\pi_\hbar(\tilde{A})$ is dense in $ {A}_{\hbar}$;
%\item[(ii)] $\lim_{\hbar\to 0}\|\pi_\hbar(\tilde{a})\|_\hbar=\|\pi_0(\tilde{a})\|_0$ for each $\tilde{a}\in \tilde{A}$.
%\end{itemize}
%Then there exists a continuous bundle of $C^*$-algebras $\cA=(I,A,\pi_\hbar)$, where $A$ is the $C^*$-algebra defined by
%$$ A:=\left\{ a \in\Pi_{\hbar} {A}_{\hbar} \,\Big| \,\lim_{\hbar \to 0}\|\pi_\hbar({a})-\pi_\hbar(\tilde{a})\|_\hbar=\|\pi_0({a})-\pi_0(\tilde{a})\|_0 \text{ where }  \tilde{a}\in \tilde{A} \right\}\,$$
%with pointwise multiplication.
% Furthermore the space of sections of $\cA$   contains $\tilde{A}$.
%\end{prop}

\begin{example}\label{ex:S2}
As an example, we consider the strict deformation quantization of Poisson manifold $\S^2$ whose Poisson bracket on $C^\infty(\S^2)$ is defined by 
\begin{align*}
\{f,g\}({\bf x}) := \sum_{a,b,c =1}^{3}\epsilon_{abc}x_c \frac{\partial f}{\partial x_a}\frac{\partial g}{\partial x_b}\:, \quad {\bf x}\in \S^{2},\:
\end{align*}
where $\epsilon_{abc}$ is the Levi-Civita symbol. To construct a continuous bundle of unital $C^*$-algebras, we set $I:=1/\N \cup \{0\}$ and ew consider the family of $C^*$-algebras 
$$ A_\hbar:=\begin{cases}
 C(\S^2) & \text{ for }\hbar=0 \\
 \textnormal{Mat}_{n+1}(\mathbb{C}) & \text{ for }\hbar\in 1/\N
  \end{cases} $$
 where $n:=1/\hbar$ and Mat$_{n+1}(\mathbb{C})$ denotes the space of $(n+1)\times (n+1)$-complex matrices. Let now set $\tilde{A}_0$ to be the algebra of polynomials in three real variables restricted to $\S^2$. Clearly, $\tilde{A}_0$ is a dense Poisson sub-algebra of $C^\infty(\S^2)$ whose Poisson bracket is defined by restricting the Poisson bracket of $\S^2$. Now let $Q_\hbar:\tilde{A}_0\to A_\hbar$ be the map defined by
 \begin{align}\label{defquan3}
Q_{\hbar}(P)& :=
 \frac{1/\hbar +1}{4\pi}\int_{S^2}P({\bf x})\, |{\bf x}\rangle\langle{\bf x}|_\frac{1}{\hbar} \,d\mu_{\bf x}\:,
\end{align}
where $d\mu_{\bf x}$ indicates the unique $SO(3)$-invariant Haar measure on ${\S}^2$ with $\int_{{\S}^2} d\mu_{\bf x} = 4\pi$  and $|{\bf x}\rangle\langle{\bf x}|_{1/\hbar}\in B(\text{Sym}^{1/\hbar}(\mathbb{C}^2))\simeq M_{1/\hbar+1}(\mathbb{C})$ is the projection onto the linear span of the unit vector
$x_{1/\hbar}$ (we refer to~\cite{VGRL18,LMV} for further details on $\text{Sym}^{1/\hbar}(\mathbb{C}^2)$). As explained in more details in the proof of~\cite[Theorem 8.1]{Lan17}, the  $C^*$-algebra $\tilde A$ consisting of 
$$ \pi_\hbar(a):=\begin{cases} 
f & \text{for } \hbar=0\\
Q_\hbar(f)  & \text{for } \hbar \in 1/\N
\end{cases}$$
 for every $f\in C(\S^2)$ is a continuous bundle of $C^*$-algebras and $Q_\hbar$ defines a quantization map which satisfies Properties $(i)-(iv)$ of Definition~\ref{def:deformationq}.
\end{example}

\begin{rmk}
 Let us remark that the quantization maps $Q_\hbar$ constructed in Example~\ref{ex:S2} define a so-called {\it Berezin quantization}, see e.g. \cite{Lan98} and that, in physics literature, the unit vector $x_{1/\hbar}$ are called coherent spin states, see e.g~\cite{Pe72}.
% , for which \eqref{Berezinprop} typically holds as well as positivity, in that $Q_{\hbar}(f)\geq 0$ if $f\geq 0$ almost everywhere on $\S^2$. 
\end{rmk}
%%%%%%%%%%%%%%%%%%%%%%%%%%%%%%%

%%%%%%%%%%%%%%%%%%%%%%%%%%%%%%%%%%%%
\section{Products of Poisson algebras}

%Given a compact Poisson manifold $X$, it is still an open question whether $X$ admits a strict deformation quantization or not. 
Let $A$ and $B$ two Poisson commutative $*$-algebras (densely contained in commutative $C^*$-algebras $\bar{A}$ and $\bar{B}$, respectively) and assume that there exists a strict deformation quantization of $A$ and $B$ respectively.
The aim of this section is to provide a necessary and sufficient criteria for the existence of a strict deformation quantization of the {algebraic} tensor product {$A\otimes B$}. We start by showing that {$A\otimes B$} is a dense Poisson $*$-subalgebra of $\bar{A}\hotimes \bar{B}$.
%\begin{lemma}\label{lem:tensor product Poisson alg}
% Let $A$ and $B$ Poisson algebras and denote with $\{\cdot\,,\cdot\}_A$ and $\{\cdot\,,\cdot\}_B$ the respectively Poisson brackets. Then tensor product $A\otimes B$ has a natural structure of Poisson algebra given by
%$$ \{ f_1 \otimes f_2 , g_1 \otimes g_2\}_{A\otimes B} := \{f_1,g_1\}_A  \otimes f_2g_2 + f_1g_1 \otimes \{f_2,g_2\}_B\,. $$
%Furthermore, if $\phi: A \to C$ and $\varphi:B\to D$ are  Poisson maps then 
%$$ \phi \otimes \varphi (\{ f_1 \otimes f_2 , g_1 \otimes g_2 \}_{A\otimes B} ) = \{ \phi (f_1) \otimes \phi(f_2) , \varphi (g_1) \otimes \varphi (g_2)\}_{C\otimes D} $$
%for every $f_1,f_2 \in A$ and $g_1,g_2 \in B$. 
% \end{lemma}
% \begin{proof}
%
% \end{proof}
%Using the jargon of cateogry theory, Lemma~\ref{lem:tensor product Poisson alg} can be restated by saying that (\textbf{Poiss}, $\otimes$) is a monoidal category.
%
\begin{lemma}\label{lem:prod poiss alg}
Let $A$ and $B$ be dense Poisson $*$-subalgebras of commutative $C^*$-algebras $\bar{A}$ and $\bar{B}$ respectively. Then there exists a Poisson structure on  $A\otimes B$  and $A\otimes B$ is dense in $\bar A\hotimes \bar B$.
\end{lemma} 

\begin{proof}
Let $A\otimes B$ the algebraic tensor product of $A$ and $B$.
For any $f_1\otimes f_2, g_1\otimes g_2 \in A\otimes B$ the map $\{ \cdot,\cdot\}_{\otimes}$ defined  by
\begin{equation}\label{eq:Poiss bracket}
 \{ f_1 \otimes f_2 , g_1 \otimes g_2\}_{\otimes} := \{f_1,g_1\}_{A}  \otimes f_2g_2 + f_1g_1 \otimes \{f_2,g_2\}_{B}\,,
 \end{equation}
where $\{\cdot,\cdot\}_A$ and $\{\cdot,\cdot\}_B$ denotes the Poisson bracket on $A$ and $B$ respectively, is a Poisson bracket on $A\otimes B$. 

To conclude our proof we need to show that $A\otimes B$ is dense in $\bar A \hotimes \bar B$.  But this follows immediately because $A\otimes B$ is dense (in the cross norm $\|\cdot\|_\epsilon$) in $\bar A \otimes \bar B$ which is dense in $\bar A \hotimes \bar B$. 
\end{proof}

\begin{cor}\label{cor:poisson man}
Let $X$ and $Y$ be locally compact Poisson manifolds. Then there exists a Poisson structure on the manifold $X\times Y$.
\end{cor} 
\begin{proof}
Since $C_0^\infty(X)$ (resp. $C_0^\infty(Y)$) is a dense Poisson  $*$-subalgebra of $C_0(X)$ (resp. $C_0(Y)$), by Lemma~\ref{lem:prod poiss alg} it follows that $C_0^\infty(X)\otimes C_0^\infty (Y)$ is a Poisson algebras densely contained in $C_0(X) \hotimes C_0(Y)$. By \cite[Corollary B.17]{OA1}  we obtain that $C_0(X) \hotimes C_0(Y)\simeq C_0(X\times Y)$ and we can define a Poisson bracket on $C_0^\infty(X\times Y)$ by declaring
$$ \{ f, g\}_{C_0^\infty(X\times Y)} := \{f(\cdot,y),g(\cdot,y)\}_{C_0^\infty(X)} + \{f(x,\cdot),g(x,\cdot)\}_{C_0^\infty(Y)}\,.$$
 This concludes our proof.
\end{proof}

With the next theorem we shall provide a criterion for the existence of a strict deformation quantization of the algebraic tensor product $\tilde{A}_0\otimes \tilde{B}_0$, where $\tilde{A}_0$ and $\tilde{B}_0$ are assumed to admit a strict deformation quantization in the sense of Definition \ref{def:deformationq}.

\begin{thm}\label{thm:main}
Let $\tilde{A}_0$ and $\tilde{B}_0$ be Poisson $*$-algebras densely contained in commutative $C^*$-algebras $A_0$ and $B_0$ respectively and assume that $\tilde{A}_0$ and $\tilde{B}_0$ admit a strict deformation quantization in the sense of Definition~\ref{def:deformationq}. Denote with $\cA=(I,A,\pi_\hbar)$ (resp. $\cB=(I,B,\sigma_\hbar)$) the continuous bundle of $C^*$-algebras and with $Q_\hbar^A$ (resp $Q_\hbar^B$) the quantization map for $\tilde{A}_0$ (resp. for $\tilde{B}_0$). Then there exists a strict deformation quantization of $\tilde{A}_0\otimes \tilde{B}_0$ over the interval $I$  with a quantization map given by $Q_\hbar:=Q_\hbar^A\otimes Q_\hbar^B$
if and only if for every $\hbar \in I$ 
\begin{equation}\label{eq:iff}
\ker(\pi_\hbar \otimes \sigma_\hbar ) = \ker(\pi_\hbar ) \hotimes B + A \hotimes \ker(\sigma_\hbar)\,.
\end{equation}
%The quantization maps $Q_\hbar$ are given by $Q_\hbar:=Q_\hbar^A\otimes Q_\hbar^B$. Here $\hotimes$ is the injective tensor product and $\|\cdot\|_{\hbar,\epsilon}$ is the injective tensor norm on $A^X_\hbar \hotimes A_\hbar^Y$.
\end{thm}
\begin{proof}
We begin by showing that condition~\eqref{eq:iff} is a sufficient criterion.
By Lemma~\ref{lem:prod poiss alg}, $\tilde{A}_0\otimes \tilde{B}_0$ are a dense Poisson $*$-subalgebra of $A_0 \hotimes B_0$. Furthermore, if condition~\eqref{eq:iff} is satisfied then by Theorem~\ref{Archbold} the bundle $\cA\hotimes\cB$ is continuous.

Now we check that the quantization map $Q_\hbar:=Q_\hbar^A\otimes Q_\hbar^B$ satisfies properties (i)-(iv) in Definition~\ref{def:deformationq}. By linearity of $Q_\hbar$ it suffices to check this on elementary tensors.
 \begin{itemize}
 \item[(i)]  $Q_0=Q_0^A\otimes Q_0^B$ is the inclusion map  and $Q_{\hbar}(\mathrm{1}_{A_0\otimes B_0})=\mathrm{1}_{A_{\hbar}}\otimes \mathrm{1}_{B_{\hbar}}$ which is the unit of $A_{\hbar}\hotimes B_\hbar$.
 \item[(ii)] For every $f\otimes g \in \tilde{A}_0\otimes \tilde{B}_0$  we have 
 $$Q_{\hbar}( ({f\otimes g})^*)=Q_\hbar^A\otimes Q_\hbar^B( {f}^*\otimes  {g}^*) = Q^A_{\hbar}(f^*)\otimes Q^B_{\hbar}(g^*) = Q^A_{\hbar}(f)^*\otimes Q^B_{\hbar}(g)^* =Q_{\hbar}(f\otimes g)^*\,,$$ where we used the fact that $Q_\hbar^A$ and $Q_\hbar^B$ are quantization maps. %we set $ {f}(x)=:f(x)^*$.
\item[(iii)] Since $Q^A_{\hbar}(f)$ and $Q_\hbar^B( g)$ are continuous section of $A_\hbar$ and $B_\hbar$ respectively for any $f \in\tilde{A}_0$ and $g \in \tilde{B}_0$, then the map
\begin{align*}
&0\mapsto f\otimes g;\\
&\hbar\mapsto Q_{\hbar}(f\otimes g)=Q^A_{\hbar}(f)\otimes Q_\hbar^B( g), \  \ (\hbar\in (I\setminus\{0\}))
\end{align*}
is a continuous section of $\cA\hotimes\cB$ by construction. Indeed, the following function is continuous:
$$ \hbar \mapsto \| \pi_\hbar ( Q_{\hbar}(f\otimes g))\|_{\hbar,\epsilon}= \| \pi_\hbar (Q^A_{\hbar}(f))\|_{\hbar} \, \| \pi_\hbar (Q^B_{\hbar}(g))\|_{\hbar}\,. $$
\item[(iv)]  Each pair $f_1\otimes g_1, f_2\otimes g_2 \in \tilde{A}_0\otimes \tilde{B}_0$  one has
\begin{align*}
[Q_{\hbar}(f_1\otimes g_1),Q_{\hbar}(f_2\otimes g_2)]=&[Q^A_{\hbar}(f_1)\otimes Q^B_{\hbar}(g_1), Q^A_{\hbar}(f_2)\otimes Q^B_{\hbar}(g_2)]=\\
=&[Q^A_{\hbar}(f_1), Q^A_{\hbar}(f_2)]\otimes  Q^B_{\hbar}(g_1)Q^B_{\hbar}(g_2) \\
&+ Q^A_{\hbar}(f_2)Q^A_{\hbar}(f_1)\otimes [Q^B_{\hbar}(g_1), Q^B_{\hbar}(g_2)]
\end{align*}
and
\begin{align*}
Q_{\hbar}(\{f_1\otimes g_1,f_2\otimes g_2\}_\otimes)=& Q_{\hbar}(\{f_1,f_2\}_A \otimes g_1 g_2 + f_1f_2 \otimes \{g_1 ,g_2 \}_B)=\\
=& Q^A_{\hbar}(\{f_1,f_2\}_A) \otimes Q^B_\hbar(g_1 g_2) + Q_\hbar^A(f_1f_2) \otimes Q^B_\hbar(\{g_1 ,g_2 \}_B)
\end{align*}
where we used Equation~\eqref{eq:Poiss bracket} and  $\{\cdot,\cdot\}_A$ (resp. $\{\cdot,\cdot\}_B$) denotes the Poisson bracket on $\tilde A_0$ (resp. $\tilde B_0$).
It then follows
\begin{align*}
\|\frac{i}{\hbar}[Q_{\hbar}(f_1\otimes g_1),& Q_{\hbar}(f_2\otimes g_2)]-Q_{\hbar}(\{f_1\otimes g_1,f_2\otimes g_2\})\|_{\hbar,\epsilon} \\
&\leq \|\frac{i}{\hbar} [Q^A_{\hbar}(f_1), Q^A_{\hbar}(f_2)]\otimes  Q^B_{\hbar}(g_1)Q^B_{\hbar}(g_2) - Q^A_{\hbar}(\{f_1,f_2\}_A) \otimes Q^B_\hbar(g_1 g_2) \|_{\hbar,\epsilon} \\
& + \|\frac{i}{\hbar} Q^A_{\hbar}(f_2)Q^A_{\hbar}(f_1)\otimes [Q^B_{\hbar}(g_1), Q^B_{\hbar}(g_2)] - Q_\hbar^A(f_1f_2) \otimes Q^B_\hbar(\{g_1 ,g_2 \}_B) \|_{\hbar,\epsilon}\,.
\end{align*}
 \end{itemize}
The first term in the above inequality can be estimated as follows: 
\begin{align*}
\lim_{\hbar\to 0} \|\frac{i}{\hbar} [Q^A_{\hbar}(f_1),& Q^A_{\hbar}(f_2)]\otimes  Q^B_{\hbar}(g_1)Q^B_{\hbar}(g_2) - Q^A_{\hbar}(\{f_1,f_2\}_A) \otimes Q^B_\hbar(g_1 g_2) \|_{\hbar,\epsilon}  \\
&= \lim_{\hbar\to 0} \Big\|\Big( \frac{i}{\hbar} [Q^A_{\hbar}(f_1), Q^A_{\hbar}(f_2)] - Q^A_{\hbar}(\{f_1,f_2\}_A) \Big)\otimes  Q^B_{\hbar}(g_1)Q^B_{\hbar}(g_2) \\
&\qquad \qquad- Q^A_{\hbar}(\{f_1,f_2\}_A) \otimes  \Big( Q^B_\hbar(g_1 g_2) 
-  Q^B_\hbar(g_1)  Q^B_\hbar(g_2)  \Big)
\Big\|_{\hbar,\epsilon}
\\
 & \leq \lim_{\hbar\to 0}\|\frac{i}{\hbar} [Q^A_{\hbar}(f_1), Q^A_{\hbar}(f_2)] - Q^A_{\hbar}(\{f_1,f_2\}_A) \|_\hbar \| Q^B_{\hbar}(g_1)Q^B_{\hbar}(g_2)\|_\hbar \\
& \qquad\qquad+\| Q^A_{\hbar}(\{f_1,f_2\}_A) \|_\hbar \| Q^B_\hbar(g_1 g_2) 
-  Q^B_\hbar(g_1) Q^B_\hbar(g_2) \|_{\hbar} \to 0
\end{align*}
 where we used Equation~\eqref{eq:cross norm} together with
\begin{align*} \lim_{\hbar\to 0}\|Q_{\hbar}(f)\|_\hbar=\|f\|_{0}, \qquad \text{ and } \qquad
\lim_{\hbar\to 0}\|Q_{\hbar}(f)Q_{\hbar}(g)-Q_{\hbar}(fg)\|_\hbar=0 \,,
\end{align*}
which follows from the definition of a continuous bundle of $C^*$-algebras.
 Using a similar argument we obtain
$$ \lim_{\hbar\to 0} \|\frac{i}{\hbar} Q^A_{\hbar}(f_2)Q^A_{\hbar}(f_1)\otimes [Q^B_{\hbar}(g_1), Q^B_{\hbar}(g_2)] - Q_\hbar^A(f_1f_2) \otimes Q^B_\hbar(\{g_1 ,g_2 \}_B) \|_{\hbar,\epsilon}\to 0. $$

Since given two $C^*$-algebras, $A$ and 
$B$, $A\hotimes B$ is the smallest $C^*$-algebra containing $A\otimes B$, it follows that $\cA \hotimes \cB$ is the smallest bundle of $C^*$-algebras containing $\cA\otimes \cB$. Therefore if there exists another tensor product $\otimes_C$ which makes $A\hat\otimes_C B$ a $C^*$-algebras, $\cA \hotimes \cB$ is contained in $\cA\hat\otimes_C\cB$. Since condition~\ref{eq:iff} is a sufficient and necessary condition to make $\cA\hotimes \cB$ continuous (cf. Theorem~\ref{Archbold}), we can conclude.
\end{proof}

As explained in Section~\ref{sec:prel inj ten}, given two continuous bundle of $C^*$-algebras $\cA$ and $\cB$ over $I$, the injective tensor product $\cA\hotimes \cB$ is not continuous in general. However for $I=1/\N \cup \{0\}$, $\cA\hotimes\cB$ is a continuous bundle.
\begin{cor}\label{cor:I special}
Assume the setup of Theorem~\ref{thm:main}. If $I:=1/\N\cup \{0\}$ then there always exists a strict deformation quantization of $\tilde{A}_0\hotimes \tilde{B}_0$ over $I$.
\end{cor}
\begin{proof}
We just need to check that $\cA\hotimes \cB$ is a continuous bundle of $C^*$-algebras. But this follows from the fact that any function is continuous on $1/\N$ and $\tilde{A}_0\hotimes \tilde{B}_0$ is a nuclear $C^*$-algebras  (cf. Lemma~\ref{lem:nucl}).
\end{proof}
\begin{cor}\label{cor:XxY}
Let $X$ and $Y$ be Poisson manifold and assume there exists a strict deformation quantization of $C_0(X)$ and $C_0(Y)$ over $I=1/\N\cup \{0\}$. Then there exists a strict deformation quantization of $X\times Y$ over $I$.
\end{cor}
\begin{proof}
On account of Corollary~\ref{cor:I special}, there exists a strict deformation quantization of $C_0(X)\hotimes C_0(Y)$ which is isomorphic to $C_0(X\times Y)$ by \cite[Corollary B.17]{OA1}. To conclude our proof is enough to endow $C_0(X\times Y)$ with the Poisson structure given by Corollary~\ref{cor:poisson man}.
\end{proof}

%\subsection{Resolvent of Sch\"odinger operators}

\section{Products of KMS states}
The aim of this section is to show that given two KMS$_\beta$ states $\omega_A$ and $\omega_B$ for two $C^*$-algebras $A$ and $B$ respectively, there exists a KMS$_\beta$-state $\omega_{A\hotimes B}$ for $A\hotimes B$.  For sake of completeness let us recall the definition of a KMS$_\beta$ state.

\begin{defn}\label{KMS1}
Consider the $C^*$-dynamical system given by a $C^*$-algebra $A$ and a strongly continuous
representation $\varphi_t$ of $\R$ in the automorphism group of $A$.
A linear functional $\omega: A \to \C$ is called a KMS$_\beta$-states if the following holds true:
\begin{itemize}
\item[(1)] it is positive, i.e. $\omega(a^*a)\geq 0$ for all $a\in A$;
\item[(2)] it is normalized, i.e. $\|\omega\|:=\sup \{\omega(a) \,|\,  a\in A, \|a\|=1\}=1$;
\item[(3)] it satisfies the KMS$_\beta$-condition: 
for all $a,b\in A$ there is a holomorphic function $F_{ab}$ on the strip $S_\beta :=  \R \times i(0, \beta) \subset \C$ with a  continuous extension to $\overline{S_\beta}$ such that
$$F_{ab}(t) = \omega(a \varphi_t (b)) \qquad \text{ and } \qquad F_{ab}(t + i\beta) = \omega(\varphi_t (b)a)\,.$$
\end{itemize}
\end{defn}

\begin{thm}\label{thm:kms}
Let $\omega^A$ and $\omega^B$ be KMS$_\beta$-states for the $C^*$-dynamical systems $(A,\varphi_{t_A},\R)$ and $(B,\phi_{t_B},\R)$ respectively and denote with  $\Phi_{t,s}$ an extension of $\varphi_{t} \otimes \phi_{s}$ to an automorphism of $A\hotimes B$ such that 
\begin{equation}
\label{eq:auto ext}
\Phi_{t,s}(a \otimes b) = \varphi_{t}(a)\otimes \phi_{s}(b)
\end{equation}
 for any $a\otimes b \in A\hotimes B$. Then there exists a KMS$_\beta$ state $\omega^{A\hotimes B}$ for the $C^*$-dynamical system $(A\hotimes B,\Phi_{t,t},\R$) such that
 \begin{equation}\label{eq:states ext}
\omega^{A\hotimes B}(a \otimes b) = \omega^A(a)\, \omega^B(b)\,.
\end{equation}
\end{thm}
\begin{rmk}
Before proving our claim, let us remark that the existence of $\Phi_{t,s}$ is guaranteed by~\cite[Proposition B13]{OA1}. Furthermore, on account of~\cite[Corollary B12]{OA1}, the state $\omega^A\otimes \omega^B$ extends to a state $\omega^A\hotimes \omega^B$ on $A\hotimes B$ which satisfies Equation~\eqref{eq:states ext}. So to prove Theorem~\ref{thm:kms} it is enough to check that $\omega^{A\hotimes B}$ satisfies the KMS$_\beta$ condition.\\
Let us also remark, this theorem can be proved using modular theory.
\end{rmk}

\begin{proof}[Proof of Theorem~\ref{thm:kms}]
 We hereto denote by $S_{\beta}$ the strip associated to the KMS$_\beta$-states $\omega^A$ and $\omega^B$, and by $F^A:=F_{a_1,a_2}^A$ and $F^B:=F_{b_1,b_2}^B$ the corresponding holomorphic functions for every $a_1,a_2 \in A$, $b_1,b_2\in B$. 
 
 Consider now $d,c\in A\hotimes B$. Since $A\otimes B$ a dense $*$-subalgebra of $A\hotimes B$ there exist some sequences of $c_i\in A \otimes B$ and $d_i\in A \otimes B$ which converge in the injective tensor norm to $c$ and $d$ respectively. In particular, we may write $c_i:=\sum_{k_i}c_{k_i1}\otimes c_{k_i2}$ and $d_i:=\sum_{l_i}d_{l_i1}\otimes d_{l_i2}$, with $c_{k_i1}\otimes c_{k_i2},d_{j1}\otimes d_{j2}\in A\otimes B$. Using Equation~\eqref{eq:auto ext} and~\eqref{eq:states ext} together with the linearity of $\omega^A$ and $\omega^B$, for any $t,s \in  S_{\beta}$ it holds
\begin{align*}
\omega^{A\hotimes B}(d_i \Phi_{t,s} (c_i))=\Sigma_{k_il_i}\omega^A(d_{l_i1}\varphi_{t}(c_{k_i1}))\omega^B(d_{l_i2}\phi_{s}(c_{k_i2})).
\end{align*}
Since $\omega^A$ and $\omega^B$ are $\beta$-KMS states, it follows that
\begin{align*}
\omega^{A\hotimes B}(d_i  \Phi_{t,s}(c_i))=\Sigma_{k_il_i}F_{d_{l_i1},c_{k_i1}}^A(t)F_{d_{l_i2},c_{k_i2}}^B(s),  
\end{align*}
where $F_{d_{l_i1},c_{k_i1}}^A$ and $F_{d_{l_i2},c_{k_i2}}^B$ are holomorphic functions for any $k,l$ such that  $F_{d_{l_i1},c_{k_i1}}^A$ and $F_{d_{l_i2},c_{k_i2}}^B$ are analytic on $S_{\beta}$, continuous and bounded on $\bar{S}_{\beta}$.
  Since  for any $i$ the sums in $k_i$ and $l_i$ are finite, and the product and sum of two analytic functions remains analytic, the above expression extends to a holomorphic function $F_{d_i,c_i}$ analytic on $S_{\beta}\times S_\beta$, and bounded and continuous on the closure $\bar{S}_{\beta}\times \bar{S}_{\beta}$. This yields a sequence of holomorphic functions $F_i:=F_{d_i,c_i}$ analytic on $S_{\beta}\times S_{\beta}$, and bounded and continuous on the closure $\bar{S}_{\beta}\times \bar{S}_{\beta}$. %Note that this sequence is uniformly bounded on the boundary of $S_{\beta}\times S_{\beta}$ by $1$, since states are normalized. 
Moreover, we claim that  the sequence $(F_i)_i$ converges uniformly on the boundary of $S_{\beta}\times S_{\beta}$ to some function. 
  To verify our claim  it suffices to check this for $\mathbb{R}\times\mathbb{R}$. Hereto we take $t\times s \in\mathbb{R}\times\mathbb{R}$ and compute
\begin{align*}
\lim_i|\omega^{A\hotimes B}(d \Phi_{t,s}(c))-\omega^{A\hotimes B}(d_i \Phi_{t,s}(c_i)|^2\leq \lim_{i}||c - c_i||^2+||d-d_i||^2=0, 
\end{align*}
where we used that $\omega^{A\hotimes B}$ is a state and that $c_i$ and $d_i$ converge to $c$ and $d$, respectively.  Since the limit does not depend on $t\times s$ the convergence is uniform.
As a result of \cite[Proposition 5.3.5]{BR} the functions $F_i$ satisfy
\begin{align*}
\sup_{z\in \bar{S}_{\beta}\times \bar{S}_{\beta}}|F_i(z)|=\sup_{(t,s)\in \mathbb{R}\times\mathbb{R }}|F_i(t,s)|. 
\end{align*}
It follows that
\begin{align}
\sup_{z\in \bar{S}_{\beta}\times \bar{S}_{\beta}}|F_i(z) - F_j(z)|=\sup_{(t,s)\in \mathbb{R}\times\mathbb{R }}|F_i(t,s) - F_j(t,s)|. \label{bulkboundary2}
\end{align}
Since $(F_i)$ converges uniformly on the boundary of $S_{\beta}\times S_{\beta}$ to some function, in particular the sequence $(F_i)$ is uniformly Cauchy on the boundary. Hence, the right hand side of \eqref{bulkboundary2} tends to zero as $i,j\to\infty$. This implies that $(F_i)$ is uniformly Cauchy on $\bar{S}_{\beta}\times \bar{S}_{\beta}$ and hence the sequence $(F_i)$ converges uniformly to some continuous function $F:=F_{d,c}$ on $\bar{S}_{\beta}\times \bar{S}_{\beta}$.  In particular,  the sequence $(F_i)$ also converges uniformly to $F$ on every compact subset of $S_{\beta}\times S_{\beta}$, so $F$ is analytic on $S_{\beta}\times S_{\beta}$ by \cite[Proposition 3]{Dav}. We conclude that the limiting function $F$ is analytic on $S_{\beta}\times S_{\beta}$ and continuous and bounded on $\bar{S}_{\beta}\times \bar{S}_{\beta}$. Restricting to the diagonal, i.e. $t=s$, this function satisfies 
\begin{align*}
F_{d,c}(t)=\omega^{A\hotimes B}(d \Phi_{t,t}(c)).
\end{align*}
By a similar argument as above one can show that it holds also
\begin{align*}
F_{d,c}(t +i\beta)=\omega^{A\hotimes B}(\Phi_{t,t}(c)d)\,.
\end{align*}
This conclude our proof.
\end{proof}
As a direct consequence of Theorem~\ref{thm:main} and Theorem~\ref{thm:kms} we get the following result.
\begin{cor}
Assume the setup of Theorem~\ref{thm:main} and Theorem~\ref{thm:kms}. Let $\omega^A_\hbar$ and $\omega^B_\hbar$ be a sequence of (KMS$_\beta$-)states for $A_\hbar:=\pi_\hbar(A)$ and $B_\hbar:=\pi_\hbar(B)$. If $\omega^A_\hbar$ and $\omega^B_\hbar$ admit a classical limit, i.e. for every $f\in \tilde{A}_0$ and $g\in\tilde{B}_0$ there exist the limits
$$ \omega_0^{A}(f)= \lim_{\hbar\to 0} \omega_\hbar^A(Q^A_\hbar(f)) \qquad \text{and} \qquad \omega_0^{B}(g)= \lim_{\hbar\to 0} \omega_\hbar^A(Q^B_\hbar(g)) 
$$
then the sequence of (KMS$_\beta$-)state $\omega_\hbar^{A\hotimes B}$ has a classical limit given by
$$ \omega_0^{A\hotimes B}( f\otimes g)=\lim_{\hbar\to 0} \omega_\hbar^{A\hotimes B}(Q_\hbar(f\otimes g)) \,. $$
\end{cor}

%%%%%%%%%%%%%%%%%%%%%%%%%%%%%%%%%%%%
\section{Applications}

\subsection{Spin systems} 
In this section we show how quantum spin systems arise from classical spin systems using our quantization formalism. 

In Example \ref{ex:S2} we have seen how a single sphere $\S^2$ is quantized using quantization maps defined by Equation \eqref{defquan3}. The fibers of the continuous bundle of $C^*$-algebras are given by
$$ A_{\hbar}:=\begin{cases}
 C(\S^2) & \text{ for } \hbar =0 \\
 \textnormal{Mat}_{2 J+1}(\mathbb{C}) & \text{ for } J:=1/\hbar\in \N
  \end{cases} $$
  where $J$ plays the role of the inverse semi-classical parameter $\hbar$. 
As notice first by Lieb in~\cite{Lieb}, and independently in~\cite{MV,Ven20}, the spin operators can be obtained using the quantization map $Q_{1/J}$
\begin{equation}\label{eq:table}
\begin{aligned}
(J+1)\cos{(\theta)} & \mapsto S_z \\
(J+1)\sin{(\theta)}\cos{(\phi)} & \mapsto S_x\\
(J+1)\sin{(\theta)}\sin{(\phi)} & \mapsto S_y
\end{aligned}
\end{equation}
where $(\theta,\phi)$ (resp $(x,y,z)$ ) are spherical (resp. cartesian) coordinates on $\S^2$. As usual  $S_x,S_y,S_z$ can be understood as a (unitary finite dimensional) irreducible representation of the Lie algebra $\mathfrak{su}(2)$  on the Hilbert space $\C^{2J+1}$. Furthermore these operators satisfy $[S_x,S_y]=iS_z$ cyclically. Here the number $J$ is also called the {\em spin} of the given representation.\medskip

A general classical spin system is typically defined as a polynomial on the cartesian product of say $d$ spheres $\S^2$, denoted by $\times_d \S^2$, where $d$ indicates the number of classical spins.
Therefore the classical algebra on which classical spin systems are defined is $C(\times_d\,\S^2)$ or equivalently $C(\S^2)^{\otimes_{\epsilon} d }$ (see Example \ref{ex:inject tensor prod}).
As a by-product of Theorem~\ref{thm:main}, the quantization maps are given  by linear extension of the following map 
\begin{equation}
\begin{aligned}\label{defquan5}
&Q_{1/J}^{(d)}:\tilde{A}_0^{\otimes_{\epsilon} d }\to \underbrace{M_{2J+1}(\mathbb{C})\otimes\cdot\cdot\cdot\otimes M_{2J+1}(\mathbb{C})}_{d \ times};\\
&Q_{1/J}^{(d)}(f_1,...,f_d)=\underbrace{Q_{1/J}^{(1)}(f_1)\otimes\cdot\cdot\cdot\otimes Q_{1/J}^{(1)}(f_d)}_{d \ times},
\end{aligned}
\end{equation}
where $Q_{1/J}^{(1)}$ is given by \eqref{defquan3}, and $\tilde{A}_0$ the dense subalgebra of $C(\S^2)$ given by polynomials in three real variables restricted to the sphere $\S^2$. Keeping this in mind, we now provide three illustrating examples where quantization theory and spin systems come together.

\paragraph{The Ising model}
We consider the {\em classical Ising model} in a transverse magnetic field $B$. The corresponding function $h^{Is}\in C(\times_d \S^2)$ is defined by
\begin{align*}
h^{Is}(e_1,...,e_d)=-\sum_{j=1}^{d-1}z_iz_{j+1} - B\sum_{j=1}^{d}x_j, \ \ (e_j=(x_j,y_j, z_j) \in \S^2, \ \ j=1,...,d).
\end{align*}
 Employing the identification $C(\times_d\S^2)\simeq C(\S^2)^{\otimes_{\epsilon} d }$, we obtain 
\begin{align*}
&h^{Is}:=-\sum_{j=1}^{d-1}h_{z_j}\otimes h_{z_{j+1}} \otimes 1_{\S^2}\otimes\cdot\cdot\cdot\otimes 1_{\S^2}-B\sum_{j=1}^d h_{x_j}\otimes 1_{\S^2}\otimes \cdot\cdot\cdot\otimes 1_{\S^2}, 
\end{align*}
where each $h_z,h_x\in C(\S^2)$ are given respectively by $h_z(e_j)=z_j$ and $h_x(e_j)=x_j$ for all $j=1,...,d$.\\

In view of \eqref{eq:table}, we see that the coordinate functions $(J+1)x_i$ are mapped to $S_i$ where $i=x,y,z$. Analogously to the work done in \cite{Lieb} let us now replace these coordinates $e_j$ by $(J+1)e_j$. We then apply our quantization maps \eqref{defquan5} to this function. It not difficult to see that this image yields the following operator
\begin{align*}
H_d^{Is}=-\sum_{j=1}^{d-1}S_z(j)S_z(j+1) - B\sum_{j=1}^{d}S_x(j),
\end{align*}
where the operators $S_x(j)$ and $S_z(j)$ act as the operators $S_x$ and $S_z$ on $\mathcal{H}_j=\mathbb{C}^{2J+1}$ and as the unit matrix $1_{2J+1}$ elsewhere. This operator exactly corresponds to the {\em quantum Ising model} of $d$ immobile spin particles each with total angular momentum $J$ under a ferromagnetic coupling, defined on the Hilbert space $\mathcal{H}^d=\bigotimes_{j=1}^d\mathcal{H}_j$, with $\mathcal{H}_j=\mathbb{C}^{2J+1}$. 
%Using the property $Q_{1/J}((J+1)h_x)=\frac{2J+1}{4\pi}\int_{S^2}f_i(\Omega)|\Omega\rangle\langle\Omega|=S_x$ and  $Q_1((J+1)h_z)=S_z$ (see \cite{Lieb,Pe72,Ven20} for details), we replace each classical vector $e_j\in\S^2$ by $(J+1)e_j$ so that by definition of the quantization maps \eqref{defquan3} (for $\hbar=2J+1$) it now follows that 
%It is not difficult to see that the function $f:=\sum_{x=1}^{N-1}f_3(x)\otimes f_3(x+1)\otimes 1_2(x+2)\cdot\cdot\cdot\otimes 1_2(x+N-1)+ B\sigma_1(x)\otimes 1_2(x+1) \cdot\cdot\cdot\otimes 1_2(x+N-1) $ with $f_3(x,y,z)=z$ and $f_1(x,y,z)=x$ satisfies
Hence, 
\begin{align*}
Q_{1/J}^{(d)}(h_J^{Is})=H_d^{Is},
\end{align*}
where $h_J^{Is}$ is defined on the scaled vectors $(J+1)e_j$. Note that the operator $H_d^{Is}$ clearly depends on $J$ since it is defined on the Hilbert space $\mathcal{H}^d=\bigotimes_{j=1}^d\mathbb{C}^{2J+1}$.
This shows the interplay between on the one hand the classical symbol on a product of spheres and on the other hand the quantum Hamiltonian describing the quantum Ising model.
\paragraph{The Heisenberg model}
We consider the classical Heisenberg spin model $h^{Hei}$ on $\times_d \S^2$ defined by 
\begin{align*}
h^{Hei}(e_1,...,e_d):=-\sum_{j=1}^{d-1}x_ix_{i+1}+y_iy_{i+1}+z_iz_{i+1}.
\end{align*}
Applying the quantization maps \eqref{defquan5} to  $h^{Hei}$ we obtain by a similar argument as in the previous example $Q_{1/J}^{(d)}(h_J^{Hei})=H_d^{Hei}$, where the operator $H_d^{Hei}$ denotes the {\em quantum Heisenberg model} on the Hilbert space $\mathcal{H}^d=\bigotimes_{j=1}^d\mathbb{C}^{2J+1}$,
\begin{align*}
H_d^{Hei}=-\sum_{j=1}^{d-1}S_j\cdot S_{j+1},
\end{align*}
with each of the operators in $S_j=(S^x_j,S^y_j,S^z_j)$ acting on the Hilbert space $\mathcal{H}_J=\mathbb{C}^{2J+1}$ and as the identity elsewhere. As before, note that the function $h_J^{Hei}$ is defined on the vectors $(J+1)e_j$.
\paragraph{The Curie-Weiss model}
We stress that also mean-field quantum spin systems can me modeled using our this theory. In this case, we take the $d$-fold tensor product of e.g. the algebra $M_2(\mathbb{C})$ with itself.
A typical example is the quantum Curie-Weiss model whose Hamiltonian is given by
\begin{align*}
H_d^{CW}=-\frac{1}{2d}\sum_{i,j=1}^{d}\sigma_3(j)\sigma_3(i) - B\sum_{j=1}^{d}\sigma_x(j),
\end{align*}
with again $B$ the magnetic field. Such models share the property that they leave the symmetric subspace $\text{Sym}^d(\mathbb{C}^2)\subset\bigotimes_{i=1}^dM_2(\mathbb{C})$ of dimension $d+1$ invariant \cite{VGRL18,MV}. Therefore, one can restrict such Hamiltonians to $\text{Sym}^d(\mathbb{C}^2)$. In this setting the restricted operator acts on the Hilbert space $\mathbb{C}^{d+1}$, and the parameter $d$ now plays the role of the spin $2J$ as explained in the beginning of this section.
It has been shown \cite{MV,Ven20} that the polynomial function on the single sphere $\S^2$
\begin{align*}
h_0^{CW}(\theta,\phi)=-(\frac{1}{2}\cos(\theta)^2+B\sin(\theta)\cos(\phi)); \ \ (\theta\in [0,  \pi], \phi\in [0, 2\pi))
\end{align*}
modulo and error of $O(1/d)$ quantizes the quantum (restricted) Curie-Weiss model under the map \eqref{defquan3}. Therefore, also in this case we recover the correspondence between the classical function on $\S^2$ and the (restricted) quantum mean field Hamiltonian.
\\\medskip

\begin{rmk}
As a result of the properties of the continuous bundle of $C^*$-algebras in all these examples it may be clear that in the classical limit $J\to\infty$ the norm of the quantum Hamiltonians correspond to the supremum norm of the corresponding classical functions, in the sense that 
\begin{align*}
\lim_{J\to\infty}\|H_d^{Quantum}\|_{J}=\|h_d^{classical}\|_{0}
\end{align*}
Of course, in view of Equation~\eqref{eq:table}, one should rescale the operators $S_x,S_y,S_z$ appearing in the quantum Hamiltonians by a factor $1/(J+1)$ in order to make the above limit existing.
\end{rmk}
\begin{rmk}
We underline that the strict deformation quantization of the $d$-fold tensor product of $\S^2$ with itself provides a new perspective in order to study the thermodynamic limit (i.e. $d\to\infty$) and classical limit (i.e. $J\to\infty$) of the spin system in question. The properties of the quantization maps can be extremely useful in order to study the above mentioned limits of for example the free energy, the possible convergence of Gibbs states, or for (algebraic) ground states induced by eigenvectors \cite{Lieb,Ven20} as also explain in the introduction.  Indeed, in a slightly different context Lieb \cite{Lieb} implicitly used the properties of the quantization maps \eqref{defquan3} and \eqref{defquan5} in order to prove the existence of such limits. 
\end{rmk}

%%%%%%%%%%%%%%%%%%%%%%%%%%%%%%%%%%%%%%%%%%%%%%%%%%
\subsection{The resolvent algebra}
In this section we shall show that the resolvent of Sch\"odinger operators for non-interacting particle system can be given in terms of an integral of the tensor product of quantization maps. To achieve our goal, we shall benefit from~\cite{Buch,TvN}.

Let $(X, \sigma)$ be a symplectic vector space  admitting a complex structure  and denote be $C_\mathcal{R}(X)$ the  commutative $C^*$-algebra  of functions on $(X, \sigma)$. Similar to the case of the (non-commutative) resolvent algebra $\mathcal{R}(X, \sigma)$ of Buchholz and Grundling (cf. \cite{Buch}), the algebra  $C_\mathcal{R}(X)$ is  the $C^*$-subalgebra of $C_b(X)$ ( the algebra of continuous functions on $X$ that are bounded with respect to the supremum norm) generated by the functions
\begin{align*}
h_x^\lambda(y)=1/(i\lambda - x\cdot y),
\end{align*}
for $x\in X$ and $\lambda\in\mathbb{R}\setminus \{0\}$. The inner product $\cdot$  gives rise to a norm $||\cdot||$ and a topology (the
standard ones for real pre-Hilbert spaces $X$), making $h_x^\lambda$ a continuous function. We now define the space $\mathcal{S}_{\mathcal{R}}(X)\subset C_\mathcal{R}(X)$ consisting of so-called levees $g\circ p_x$
\begin{align*}
\mathcal{S}_{\mathcal{R}}(X)=\text{span}\{ g\circ p_x \ \ \text{levee}\ |\ g\in \mathcal{S}(\text{ran}(P) \},
\end{align*}
where a levee $f:X\to\mathbb{C}$ is a composition $f = g \circ P$ of some finite dimensional projection $P$ and some function $g \in C_0(\text{ran}(P))$. As shown  in \cite[Proposition 2.4]{TvN}  $\mathcal{S}_{\mathcal{R}}(X)$ is a dense $*$-Poisson subalgebra of $C_\mathcal{R}(X)$.

Now let us denote the resolvent algebra by $\mathcal{R}(X,\sigma)$. This is the $C^*$-subalgebra of $B(\mathcal{F}(\bar{X}))$ generated by the resolvents $R(\lambda, x):=(i\lambda - \varphi(x))^{-1}$ for $\lambda\in\mathbb{R}\setminus\{0\}$ and $x\in X$, where $\mathcal F(\bar{X})$ denoted the bosonic Fock space (symmetric Hilbert space) of the completion of $X$ with respect to its complex inner product. It can be shown that the fibers $A_0:=C_\mathcal{R}(X)$ ($\hbar=0$) and the constant fiber $A_{\hbar}=\mathcal{R}(X,\sigma)$ above  $\hbar\neq 0$ entail a continuous bundle of $C^*$-algebras over $I:=[0,\infty)$. In \cite[Theorem 3.7]{TvN} van Nuland showed that there exists a strict deformation quantization of the commutative resolvent algebra $A_0=C_\mathcal{R}(X)$ over base space $I=[0,\infty)$ with non-zero fibers given by the (non-commutative) resolvent algebra $A_{\hbar}=\mathcal{R}(X,\sigma)$. The corresponding quantization maps (denoted by $Q_{\hbar}^W$) are defined in terms of Weyl-quantization on the dense Poisson subalgebra $\mathcal{S}_{\mathcal{R}}(X)\subset C_\mathcal{R}(X)=A_0$.  Furthermore, these maps are surjective.

Since $A_0:=C_\mathcal{R}(X)$ and the resolvent algebra $A_{\hbar}=\mathcal{R}(X,\sigma)$ are nuclear $C^*$-algebras (see e.g. \cite[Proposition 3.4]{Buch2}),
there exists a strict deformation quantization of $C_\mathcal{R}(X)\otimes C_\mathcal{R}(X)$ (cf. Theorem~\ref{thm:main}). In particular,  the quantization maps are defined on the dense Poisson algebra  $\mathcal{S}_{\mathcal{R}}(X)\otimes \mathcal{S}_{\mathcal{R}}(X)\subset A_0\times A_0$.

\paragraph{Schr\"{o}dinger operators affiliated with the resolvent algebra}
From now on, we set $X=\mathbb{R}^2$ with its standard symplectic form $\sigma$ and work in the Schr\"{o}dinger representation $\pi_0$ of $\mathcal{R}(\R^2,\sigma)$.  We denote by $Q,P$
the canonical position and momentum operators in the Schr\"{o}dinger representation. Let $H=H(P,Q)$ be a self-adjoint operator. When its resolvent is contained in $\pi_0(\mathcal{R}(\mathbb{R}^2,\sigma))$ we may consider its preimage
\begin{align}
\tilde{R}_{H}(\lambda)=\pi_0^{-1}((i\lambda-H)^{-1}), \ \ (\lambda\in \mathbb{R}\setminus\{0\}),\label{aff}
\end{align}
as long as $\lambda$ is not in the spectrum of $H$. We then say that $H$ is {\bf affiliated} with $\cR(\R^2,\sigma)$. Since $\R^2$ is finite dimensional, Equation \eqref{aff}  holds for Schr\"{o}dinger operators with compact resolvent or for Schr\"{o}dinger operators with potential $V\in C_0(\mathbb{R})$ \cite[Proposition 6.2]{Buch}.

\paragraph{Many particle systems}
We consider ($\hbar$-dependent) Schr\"{o}dinger operators $H_i \ (i=1,...,N)$ each densely defined on some Hilbert space $\mathcal{H}_i$ and affiliated with $\cR(\R^2,\sigma)$. We then consider the tensor product of these operators
\begin{align}
H:=H_1\otimes 1_2\otimes\cdot\cdot\cdot\otimes 1_N + 1_1\otimes H_2\otimes \cdot\cdot\cdot\otimes 1_N + ... + 1_1\otimes1_2\otimes \cdot\cdot\cdot\otimes H_N,
\end{align}
where $1_i$ denotes the identity operator on $\mathcal{H}_i$ for $i=1,...,N$. One can extend the operator $H$ to a densely defined self-adjoint operator on $\mathcal{H}=\bigotimes_{i=1}^N\mathcal{H}_i$. By construction the operators $H_i$ now viewed as operators on $\mathcal{H}$ commute. The operator $H$ therefore describes a system of $N$ non-interacting particles. To simplify matters, let us restrict to the case when $N=2$ and let us assume that the spectra of $H_1$ and $H_2$ are bounded from below. 
It can then be shown that the resolvent of $H$ is given as a (operator valued) function of $H_2$ in terms of a Dunford integral \cite{LM}, using the fact that $R_{1}=1_1\otimes R_2$ obviously commutes with  $R_{2}=R_1\otimes 1_2$.
Concretely, this means that for any $\lambda$ in the set $\rho(H)\bigcap_{i=1}^2\rho(H_i)$ (where $\rho$ denotes the resolvent), we have
\begin{align}
R_H(\lambda)=\lim_{k\to\infty} \frac{1}{2\pi i}\int_{\Gamma_k}dz(z+\lambda +H_1)^{-1}(z - H_1)^{-1}, \label{res1}
\end{align}
where $\Gamma_k$ is a suitable contour crossing the real axis in some point $x_k\in\mathbb{R}$ where $x_k$ increasing towards infinity as $k\to\infty$. We can rewrite \eqref{res1} as 
\begin{align*}
R_H(\lambda)=\lim_{k\to\infty} \frac{1}{2\pi i}\int_{\Gamma_k}dzR_{1}(z+\lambda)R_{2}(z),  
\end{align*}
where $R_1$ and $R_2$ denote the resolvent of $-H_1$ and $H_2$, respectively. Since each of them is affiliated with $\mathcal{R}(\mathbb{R}^2,\sigma)$ we can consider their preimages under $\pi_0$ which we denote by $\tilde{R}_1$ and $\tilde{R}_2$. Since $\pi_0$ is a faithful representation we obtain
\begin{align*}
\tilde{R}_H(\lambda)=\lim_{k\to\infty} \frac{1}{2\pi i}\int_{\Gamma_k}dz\tilde{R}_{1}(z+\lambda)\tilde{R}_{2}(z).  
\end{align*}
The previous results in this section now imply the existence of two functions $f_1^{z+\lambda},f_2^{z}\in \mathcal{C}_{\mathcal{R}}(\R^2)$ such that
\begin{align*}
&\tilde{R}_{1}(z+\lambda)=Q_{\hbar}^W(f_1^{z+\lambda})\otimes 1_2;\\
&\tilde{R}_{2}(z)=1_1\otimes Q_{\hbar}^W(f_2^{z}).
\end{align*}
Combining the above results yields
\begin{align*}
\tilde{R}_H(\lambda)=\lim_{k\to\infty} \frac{1}{2\pi i}\int_{\Gamma_k}dzQ_{\hbar}^W(f_1^{z+\lambda})\otimes Q_{\hbar}^W(f_2^{z}). 
\end{align*}
This implies that the resolvent of Sch\"odinger operators for non-interacting particle system (as defined above) can be given in terms of an integral of the tensor product of quantization maps, quantizing functions in the commutative resolvent algebra.

\end{document}